\documentclass[11pt]{article}
\usepackage{geometry}
\usepackage{etex}
\usepackage{amssymb,amsmath,amscd}
\usepackage{amsthm}
\usepackage[usenames,dvipsnames,svgnames,table]{xcolor}
\usepackage{tikz}
\usetikzlibrary{arrows,shapes,positioning}
\usepackage{tikz-cd}
\usepackage{subcaption}
\usepackage{mathtools}
\usepackage{dsfont}
\usepackage[numbers]{natbib}
\usepackage{hyperref}
\usepackage{url}
\usepackage{enumerate}

\newtheorem{proposition}{Proposition}
\newtheorem{theorem}{Theorem}
\newtheorem{lemma}{Lemma}

\newtheorem{definition}{Definition}

\theoremstyle{remark}
\newtheorem{remark}{Remark}
\newtheorem{example}{Example}

\newcommand{\e}{\mathrm{e}}
\newcommand{\tf}{\tilde{f}}
\newcommand{\A}{\mathcal{A}}
\newcommand{\B}{\mathbb{B}}
\newcommand{\G}{\mathcal{G}}
\newcommand{\tG}{\tilde{G}}
\renewcommand{\b}{b}
\newcommand{\X}{X}
\newcommand{\Y}{\B^m}
\newcommand{\sign}{\mathrm{sign}}
\newcommand*\circled[1]{\tikz[baseline=(char.base)]{
                \node[shape=circle,draw,inner sep=1.5pt] (char) {#1};}}
\newcommand{\NN}[1]{\{1,\dots, #1\}}

\begin{document}

\title{On the conversion of multivalued to Boolean dynamics}
\author{Elisa Tonello \\
       School of Mathematical Sciences \\
       University of Nottingham, Nottingham, NG7 2RD}
\date{}

\maketitle

\begin{abstract}
Results and tools on discrete interaction networks are often concerned with Boolean variables,
whereas considering more than two levels is sometimes useful.
Multivalued networks can be converted to partial Boolean maps, in a way that preserves the asynchronous dynamics.
We investigate the problem of extending these maps to non-admissible states, i.e.\ states that do not have a multivalued counterpart.
We observe that attractors are preserved if a stepwise version of the original function is considered for conversion.
Different extensions of the Boolean conversion affect the structure of the interaction graphs in different ways.
A particular technique for extending the partial Boolean conversion is identified, that ensures that feedback cycles are preserved.
This property, combined with the conservation of the asymptotic behaviour,
can prove useful for the application of results and analyses defined in the
Boolean setting to multivalued networks, and vice versa.
As a first application, by considering the conversion of a known example for the discrete multivalued
case, we create a Boolean map showing that the existence of a cyclic attractor and the absence of fixed points
are compatible with the absence of local negative cycles.
We then state a multivalued version of a result connecting mirror states and local feedback cycles.
\end{abstract}

\section{Introduction}
Boolean networks model the dynamics of components taking two possible values, $0$ and $1$, and have a wide
range of applications, having been used for decades to investigate, for instance, biological networks~\cite{glass1973logical,thomas1973boolean,thomas1981relation}.
In some cases, discrete models that allow for more than two values are necessary to capture fundamental properties of the system behaviour~\cite{didier2011mapping,thomas1991regulatory}.
When discrete maps are used to model biological networks, the asynchronous update scheme is often considered,
so that only the value of one component can change at each iteration.
For each state of the system, a local interaction graph can be associated to the discrete network,
which defines how the update rule for the system depends on the variables.
The global interaction graph of the network is then defined by considering all possible local interactions.

Even though the literature on multivalued networks has been expanding in the last few years
(e.g.,~\cite{richard2007necessary,richard2008extension,richard2010negative,chaouiya2011petri}),
many tools focus on the Boolean setting.
To enable a possible application of these techniques to the more general multivalued case,
a mapping of discrete maps to Boolean has been considered,
by defining a Boolean variable for each positive value of each component,
thus embedding multivalued states into a higher dimensional Boolean space~\cite{van1979deal,didier2011mapping}.
The Boolean maps obtained are however defined on a subset of the Boolean configurations,
whereas results and tools on Boolean networks usually require maps to be defined on all Boolean states.
Here we consider possible ways of extending the partial Boolean conversions of multivalued maps
to the states called ``non-admissible''~\cite{didier2011mapping},
i.e.~Boolean states that do not have a discrete counterpart.
We show that, if the stepwise version of the discrete map is converted,
then the attractors of the asynchronous dynamics are preserved (Propositions~\ref{prop:A_trap_domain} and~\ref{prop:nonadm_to_adm}).
In addition, we identify a particular extension of the partial Boolean map that has the following property:
the interaction graph of the Boolean conversion contains a cycle if and only if so does the interaction graph of the original discrete map (Theorems~\ref{thm:circuits_f_to_F} and~\ref{thm:circuits_F_to_f}).

The focus on the preservation of cycles is motivated by the numerous results that have established how the presence or absence of cycles
can determine the possible asymptotic behaviours of Boolean and discrete networks.
If the local interaction graphs do not admit any cycle, then the map has a unique fixed point,
a result due to Shih and Dong~\cite{shih2005combinatorial} for the Boolean case,
and later generalised to the multivalued case by Richard~\cite{richard2008extension}.
Versions of the first and second conjectures of R. Thomas~\cite{thomas1981relation} have been proved: the existence of multiple attractors requires the presence of a local positive cycle~\cite{remy2008graphic,richard2007necessary},
whereas cyclic attractors are precluded in absence of negative cycles in the global interaction graph~\cite{remy2008graphic,richard2010negative}.
The question of whether a local negative cycle is necessary for the existence of a cyclic attractor
has been answered first with a counterexample in the multivalued case by Richard~\cite{richard2010negative}.
Ruet~\cite{ruet2016local} presented alternative proofs and generalisations for the results connecting
attractors and local cycles in the Boolean case, and showed that
local negative cycles, if they exist, might be found away from the cyclic attractors.
The first examples of maps with a cyclic attractor and no local negative cycle in the Boolean case have been exhibited by Ruet~\cite{ruet2017negative}, for $n\geq 7$.
Here we apply the conversion method from multivalued to Boolean to the multivalued counterexample of Richard,
and exhibit a simple example of map on $6$ Boolean components
having no fixed points and no local negative cycles (Section~\ref{sec:ex6}).
An alternative Boolean conversion for the same map was recently presented by Faur{\'e} and Kaji~\cite{faure2018circuit};
a comparison to their conversion method is given in Section~\ref{sec:comparison}.
As a second application of the conversion, we give a multivalued version of the result of Ruet~\cite{ruet2016local} on the existence of local cycles
in presence of mirror pairs (Theorem~\ref{thm:mirror_multi}).

Finally, in Section~\ref{sec:circuits_admissible}, we study additional properties of the conversion on the admissible states,
considering in particular other interaction graphs introduced in the literature to investigate multistationarity and oscillations in discrete systems.

\section{Background}\label{sec:background}

We consider a system of $n$ interacting components, with the $i^{th}$ component taking values in a finite interval of integers $X_i=\{0,\dots,m_i\}$,
and write $\X=\prod_{i=1}^n\X_i$.
We denote by $\e^i$ the element of $\X$ with $\e^i_i=1$ and $\e^i_j=0$ for $j=1,\dots,n$, $j\neq i$, and write $\B=\{0,1\}$.
For $x\in\B^n$ and $I\subseteq\NN{n}$, $\bar{x}^I$ will denote the element of $\B^n$ with $(\bar{x}^I)_i=1-x_i$ for $i\in I$ and $(\bar{x}^I)_i=x_i$ for $i\notin I$.
If $I$ consists of one element $I=\{i\}$, we will write $\bar{x}^i$ for $\bar{x}^{\{i\}}$.
In this work we will consider maps from some set $Y\subseteq X$ to itself, and call the elements of $Y$ the states of the system.
Given $f\colon Y\to Y$, we will write $f_i$ for the map from $Y$ to $X_i$ defined by $f_i(x)=x_i$ for $x\in Y$.

The dynamics defined by the iteration of a map $f\colon Y\to Y$, $Y\subseteq X$, is called \emph{synchronous dynamics}.
In application areas such as gene regulatory networks, the graph called \emph{asynchronous dynamics}
or \emph{asynchronous state transition graph} $AD_f$ is often investigated.
This is the graph with vertex set $Y$, and edge set $\{(x,y) | x,y \in Y, y=x+s\e^i, s = \sign(f_i(x)-x_i), s\neq 0, i=1,\dots,n\}$,
where we define $\sign(a)=0$ for $a=0$, and $\sign(a)=a/|a|$ for $a\in\mathbb{Z}$, $a\neq 0$.

Of particular interest are the asymptotic behaviours of asynchronous dynamics, and their relationships
to the interaction structure encoded by the map $f$.
A non-empty subset $A\subseteq Y$ is a \emph{trap set} for $AD_f$ if, for every edge $(x,y)\in AD_f$, $x\in A$ implies $y\in A$.
The minimal trap sets with respect to the inclusion are called \emph{attractors} for the asynchronous dynamics.
Attractors consisting of a single state are called \emph{stable states} or \emph{fixed points},
and the other attractors are called \emph{cyclic}.

\begin{definition}\label{def:int_graph}
  Given a map $f: Y\to Y$, with $Y\subseteq X$, the \emph{(local) interaction graph} of $f$ at $x \in Y$ is the finite labelled directed graph $G_f(x)$ with vertex set $\{1, \dots, n\}$,
  and an edge from $j$ to $i$,
  with label $s = s_1 (\sign(f_i(x + s_1 \e^j) - f_i(x))$, for $s_1 \in {\{-1, 1\}}$ and $x + s_1 \e^j \in Y$, whenever $s \neq 0$.
  We will call $s$, $s_1$ and $j$ the \emph{sign}, \emph{variation} and \emph{direction} of the edge, respectively.

  The \emph{global interaction graph} $G_f$ of $f$ is the graph with vertex set $\{1,\dots,n\}$
  and an edge from $j$ to $i$ with sign $s$ if, for some $x\in Y$, the graph $G_f(x)$
  has an edge from $j$ to $i$ with sign $s$.
\end{definition}

The cycles we consider in this work are all elementary, unless otherwise stated.
Cycles of length $1$ will be called \emph{loops}.
The sign of a cycle is defined as the product of the signs of its edges.
For an edge $(j,i)$ of sign $s$ in an interaction graph will also use the notations $j\to i$ and $j\xrightarrow{s} i$.

In the Boolean case, the asynchronous dynamics $AD_f$ uniquely determines the map $f$; this is not the case in the multivalued setting.
Given a map $f\colon\X\to\X$, we will consider a particular map $\tf\colon\X\to\X$,
called the \emph{stepwise} map associated to $f$, defined by
\begin{equation*}
  \tf_i(x) = x_i + \sign(f_i(x)-x_i), \text{ for } x\in\X, \ i=1,\dots,n.
\end{equation*}
A map $f$ will be called \emph{stepwise} if it coincides with its stepwise version $\tf$.
It is easy to see that $f$ and $\tf$ admit the same asynchronous dynamics,
and that a stepwise map is uniquely determined by its asynchronous dynamics.

Maps admitting the same asynchronous dynamics can have different interaction graphs,
as shown in the following example.

\begin{example}
  Consider the map $f\colon\{0,1,2\}^2\to\{0,1,2\}^2$ defined by $f(0,1)=f(0,2)=(2,2)$, and $f(x)=(1,2)$ for $x\neq (0,1)$, $x\neq (0,2)$.
  The graph $G_f$ admits a negative loop $1\to 1$ and a positive edge $2\to 1$.
  The stepwise version $\tf$ of $f$ verifies $\tf(0,0)=\tf(1,0)=\tf(2,0)=(1,1)$, $\tf(x)=(1,2)$ otherwise.
  Hence the graph $G_{\tf}$ consists of a positive loop $2\to 2$, and $G_f$ and $G_{\tf}$ have no edges in common.
\end{example}

One can easily verify, however, that if a local interaction graph $G_{\tf}(x)$ of the stepwise version of a map $f$
contains an edge from $j$ to $i$ with sign $s$, with $j\neq i$, then the local interaction graph $G_f(x)$ also contains an edge
from $j$ to $i$, with sign $s$.

Tools for the analysis of discrete interaction networks often focus on the Boolean case,
and therefore a mapping from multivalued to Boolean dynamics can prove useful.
Didier et al.~\cite{didier2011mapping} considered the mapping introduced by Van Ham~\cite{van1979deal},
and showed that it is the only mapping that can preserve both the interaction structure and the dynamical properties of the system.
The mapping is defined by considering a one-to-one correspondence between the multivalued space $\X$ and a subset of the Boolean space with dimension $m=\sum_{i=1}^nm_i$.
For convenience, we will index the components of $\B^m$ using pairs of indices $I=\{(i,j)|i=1,\dots,n,j=1,\dots,m_i\}$, and set $I_i=\{(i,j)|j=1,\dots,m_i\}$ for $i=1,\dots,n$.
Define $m$ functions $\b_{i,j}:\X\to\B$ as $\b_{i,j}(x)=\chi_{[j,m_i]}(x_i)$ for $i=1,\dots,n$, $j=1,\dots,m_i$, where $\chi_A$ is the indicator function of the set $A$.
Then the map
\begin{equation*}
  \begin{aligned}
    \b\colon\X&\to\B^m\\
    x&\mapsto (\b_{1,1}(x),\b_{1,2}(x),\dots,\b_{1,m_1}(x),\b_{2,1}(x),\dots,\b_{n,m_n}(x))
  \end{aligned}
\end{equation*}
is injective. The image $\A=\b(\X)$ is called the \emph{admissible region} or the set of \emph{admissible states}.
These are the states $y\in\B^m$ such that, if $y_{i,j}=1$ for some $(i,j)\in I$, then $y_{i,h}=1$ for all $h<j$.
Given a multivalued discrete dynamics $f$ on $\X$, we call a \emph{Boolean conversion} of $f$
any map $F: \B^m \rightarrow \B^m$ defined so that the following diagram is commutative:
\begin{center}
  \begin{tikzcd}
  \X \arrow[r, "f"] \arrow[d, "\b"]
  & \X \arrow[d, "\b"] \\
  \B^m \arrow[r, "F"] & \B^m
  \end{tikzcd}
\end{center}
Didier et al.~\cite{didier2011mapping} study the map $f^b= \b \circ f \circ \b^{-1}\colon\A\to\A$,
that we call \emph{partial Boolean conversion}.
They show in particular that $\b$ defines an isomorphism between the asynchronous state transition graph of $f$
and the asynchronous state transition graph of $f^b$, which is defined on the admissible region~\citep[Proposition 4]{didier2011mapping}.
Many tools and results on Boolean networks, however, can not be applied directly to the
partial Boolean conversion $f^b$.
In the following sections we consider the problem of defining the Boolean map $F$ outside of the admissible region,
without disrupting the asymptotic behaviour and the interaction structure.

\section{Preservation of asymptotic behaviour}\label{sec:asympt}

We want to be able to study a Boolean conversion $F$ of $f$ to derive properties of the multivalued map $f$.
Ideally, the dynamics of $F$ and $f$ should have the same asymptotic behaviour.
The asynchronous dynamics of $F$ on the admissible states coincides with the asynchronous dynamics of
the partial Boolean conversion $f^b$ and is therefore isomorphic to the asynchronous dynamics of $f$.
However, when the map $f^b$ is extended to non-admissible states, the asynchronous dynamics can leave the admissible region.
Consider for instance a map $f$ on $\X=\{0,1,2\}$ that sends $0$ to $2$: any Boolean conversion $F$ of $f$ verifies $F(0,0)=(1,1)$,
and the asynchronous dynamics $AD_F$ contains a transition from the state $(0,0)$ to the admissible state $(1,0)$,
as well as a transition from the state $(0,0)$ to the non-admissible state $(0,1)$.
This problem can be avoided by considering the conversion of the stepwise version of $f$ instead, as shown in the following proposition.

\begin{proposition}\label{prop:A_trap_domain}
  Let $F\colon\Y\to\Y$ be a Boolean conversion of a stepwise map $f\colon\X\to\X$.
  Then the set of admissible states $\A$ is a trap domain for $AD_{F}$.
\end{proposition}
\begin{proof}
Let $a$ be an admissible state of $\Y$. Then $a = \b(x)$ for some $x \in \X$.
Suppose that, for some indices $i, j$, with $i \in \{1, \dots, n\}$ and $j \in \{1, \dots, m_i\}$, we have $F_{i,j}(a) \neq a_{i,j}$, or, in other words, $(a, \bar{a}^{i,j})$ is in $AD_{F}$.
We want to prove that $\bar{a}^{i,j}$ is admissible.
We show that $\bar{a}^{i,j} = \b(x + \epsilon \e^i)$, with $\epsilon = \sign(f_i(x) - x_i)$.

First observe that, since $F_{i,j}(a) \neq a_{i,j}$, we have $\bar{a}_{i,j}^{i,j} = F_{i,j}(a) = f^b_{i,j}(a) = \b_{i,j}(f(x)) = \b_{i,j}(x+\mathrm{sign}(f_i(x)-x_i)\e^i) = \b_{i,j}(x + \epsilon \e^i)$.

On the other hand, it follows from the definition of $\b$ that
\begin{equation*}
  \b(x+\e^i)=\b(x)+\e^{i,x_i+1} \text{ and } \b(x-\e^i)=\b(x)-\e^{i,x_i}.
\end{equation*}
As a consequence, we have that $x_i=j-1$ if $\epsilon=+1$, and $x_i=j$ if $\epsilon=-1$.
In both cases we find $\bar{a}^{i,j}_{k,h}=\b_{k,h}(x+\epsilon\e^i)$ for all $(k,h)\neq(i,j)$.
\end{proof}

It follows from the proposition that, if $F$ is a Boolean conversion of the stepwise version of a map $f\colon\X\to\X$,
to obtain a one-to-one correspondence between the attractors of $AD_f$ and the attractors of $AD_{F}$,
it is sufficient to ensure that, for each state in $\Y\setminus\A$, there exists a path to the set of admissible states $\A$.
If this additional condition is verified, we call the map $F$ a \emph{compatible} conversion of $f$.
The following proposition gives a sufficient condition: if the non-admissible states are mapped to the admissible,
and the map is, in a sense, stepwise also outside of the admissible region, then $F$ is a compatible conversion of $f$.

\begin{proposition}\label{prop:nonadm_to_adm}
Consider a multivalued map $f: \X \rightarrow \X$ and a Boolean conversion $F:\Y\to\Y$ that satisfies $F(x)\in\A$ and $|\sum_{j=1}^{m_i}F_{i,j}(x)-\sum_{j=1}^{m_i}x_{i,j}|\leq 1$ for all $x\in\Y$ and $i=1,\dots,n$.
Then for each $x\in\Y$ there exists a path from $x$ to $\A$ in $AD_{F}$.
\end{proposition}
\begin{proof}
We proceed by induction on the Hamming distance $d(x,\A)$. If $d(x,\A)=0$, then $x$ is in $\A$ and there is nothing to prove.
If $d(x,\A)>0$, we show that $x$ admits a successor $y\in\Y$ in the asynchronous dynamics that satisfies $d(y,\A)=d(x,\A)-1$,
and the conclusion follows by the induction hypothesis.
Since $x$ is not admissible, there exists a pair of indices $i,j$ such that $x_{i,j}=0$ and $x_{i,j+1}=1$.
Define $j_1=\min\{1\leq j\leq m_i|x_{i,j}=0\}$, $j_2=\max\{1\leq j\leq m_i|x_{i,j}=1\}$.
Observe that, if $J\subseteq I$ is such that $\bar{x}^J\in\A$, then $(i,j_1)\in J$ or $(i,j_2)\in J$.
In addition, since the distance between $\sum_{j=1}^{m_i}F_{i,j}(x)$ and $\sum_{j=1}^{m_i}x_{i,j}$ is at most $1$ and $F(x)$ is admissible,
we have $F_{i,j}(x)=1$ for all $j<j_1$, $F_{i,j}(x)=0$ for all $j>j_2$, and
\begin{equation}\label{hp:dist}
\left| \sum_{k=j_1}^{j_2}(F_{i,k}(x)-x_{i,k})\right|\leq 1.
\end{equation}
Consider first the case with $F_{i,j_2}(x)=1$, which implies $F_{i,k}(x)=1$ for all $k\leq j_2$.
Using~\eqref{hp:dist}, we find that $x_{i,k}=1$ for all $j_1<k\leq j_2$.
Hence, if $J$ is a set of indices in $I$ such that $\bar{x}^J\in\A$ and $|J|=d(x,\A)$, the cardinality of $J\cap I_i$ is exactly one,
and $y=\bar{x}^{i,j_1}$ is a successor for $x$ in $AD_{F}$ that satisfies $d(y,\A)=d(x,\A)-1$.

Suppose now that $F_{i,j_2}(x)=0$.
If $(i,j_2)$ belongs to a set of indices $J$ such that $|J|=d(x,\A)$ and $\bar{x}^J\in\A$, then $\bar{x}^{i,j_2}$ is a successor for $x$ with $d(y,\A)=d(x,\A)-1$.
Otherwise, for any set of indices $J$ with $|J|=d(x,\A)$ and $\bar{x}^J\in\A$, we must have $(i,j_1)\in J$, and,
in addition, there must be an index $h$ with $j_1<h<j_2$ such that $x_{h}=1$.
Using~\ref{hp:dist}, we find that $F_{i,j_1}(x)=1$, and $\bar{x}^{i,j_1}$ is the required successor.
\end{proof}

\section{Preservation of cycles in the interaction graphs}\label{sec:circuits}

We now move our attention to the interaction structure of $f$ and of its Boolean conversions.
We start by establishing some relationships between the edges in the interaction graphs of $f$
and the edges in the interaction graphs of $f^b$.

\begin{lemma}\label{lem:int_graph_f_to_fb}
    Consider $f\colon\X\to\X$ and $x\in\X$.
    If the graph $G_{f}(x)$ contains an edge from $j$ to $i$ with sign $s$, variation $s_1$ and direction $j$,
    then the graph $G_{f^b}(\b(x))$ contains edges with sign $s$ from the vertex $(j, x_j + \frac{s_1+1}{2}) \in I_j$
    to vertices $(i,k') \in I_i$, for all $k' \in \left] \min{\{f_i(x),f_i(x + s_1\e^j)\}}, \max{\{f_i(x),f_i(x + s_1\e^j)\}}\right]$.
\end{lemma}
\begin{proof}
    Suppose that $s = s_1 \sign(f_i(x+s_1\e^j)-f_i(x))$ for some $s, s_1 \in \{-1,1\}$, $i, j \in \{1,\dots,n\}$.
    First observe that, if $x$ and $x + s_1\e^{j}$ are in $\X$, and $y = \b(x)$, then $y + s_1\e^{j,k} = \b(x + s_1\e^j)$, with $k= x_j + \frac{s_1+1}{2}$.

    Take $k= x_j + \frac{s_1+1}{2}$, and $k' \in \left] \min{\{f_i(x),f_i(x + s_1\e^j)\}}, \max{\{f_i(x),f_i(x + s_1\e^j)\}}\right]$.
    We have
    \begin{equation*}
        \begin{aligned}
            s_1 s & = \sign(f_i(x+s_1\e^j)-f_i(x)) \\
                  & = \sign(\b_{i,k'}(f(x+s_1\e^j))-\b_{i,k'}(f(x))) = \sign(f^b_{i,k'}(y+s_1\e^{j,k}) - f^b_{i,k'}(y)),
        \end{aligned}
    \end{equation*}
    as required.
\end{proof}

\begin{lemma}\label{lem:int_graph_fb_to_f}
    Consider $f\colon\X\to\X$ and $x\in\X$.
    If the graph $G_{f^b}(\b(x))$ contains an edge from a vertex in
    $I_{j}$ to a vertex in $I_{i}$ with sign $s$, with $i, j \in \{1,\dots,n\}$,
    then the graph $G_{f}(x)$ contains an edge from $j$ to $i$ with sign $s$.
\end{lemma}
\begin{proof}
  Suppose that $y = \b(x)$ for some $x\in\X$, and that $G_{f^b}(y)$ contains an edge from $j,k$ to $i,k'$ with variation $s_1$.
  Then the state $y + s_1\e^{j,k}$ is admissible, and $y + s_1\e^{j,k} = \b(x + s_1\e^j)$. We can therefore write
  \[f^b_{i,k'}(y+s_1\e^{j,k}) - f^b_{i,k'}(y) = \b_{i,k'}(f(x+s_1\e^j))-\b_{i,k'}(f(x)),\]
  and
  \[s = s_1 \sign(f^b_{i,k'}(y+s_1\e^{j,k}) - f^b_{i,k'}(y)) = s_1 \sign(\b_{i,k'}(f(x+s_1\e^j))-\b_{i,k'}(f(x))).\]
  This implies $s = s_1 \sign(f_i(x+s_1\e^j)-f_i(x))$ as required.
\end{proof}

The lemmas show that every interaction associated to $f$ has at least one counterpart in $G_{f^b}$,
and that each interaction in $G_{f^b}$ derives from some interaction in $G_f$.
When considering the interaction graphs of extensions $F$ of the partial Boolean conversion $f^b$, the last property does not hold in general,
even for maps that are compatible conversions of $f$.
Consider for instance the stepwise multivalued map $f\colon\{0,1,2\}\to\{0,1,2\}$ defined by $f(0)=f(1)=0$, $f(2)=1$.
We can define a compatible conversion $F\colon\{0,1\}^2\to\{0,1\}^2$ by setting $F(0,0)=F(1,0)=(0,0)$, $F(1,1)=(1,0)$ and,
for instance, $F(0,1)=(1,1)$. The interaction graph $G_f$ has only a positive edge $1\to 1$, whereas the graph $G_F$ admits also a negative interaction $(1,1)\to(1,2)$.

In the following, we describe a construction that allows to derive a version of Lemma~\ref{lem:int_graph_fb_to_f} for a Boolean conversion.
The existence of a cycle in the interaction graph of a multivalued map
does not imply, in general, the existence of a cycle in the interaction graph of a compatible conversion
(see the example in Figure~\ref{fig:ex_local_not_preserved}).
We will prove this property for the particular Boolean conversion presented here.

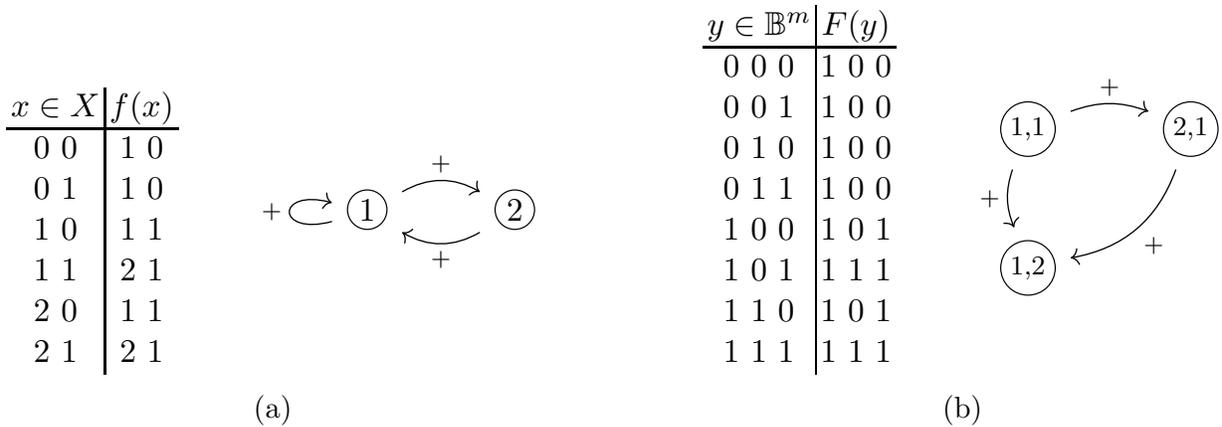
\begin{figure}
  \begin{center}
  \resizebox{\columnwidth}{!}{
\subcaptionbox{\label{fig:ex_local_not_preserved_a}}{
\tabcolsep=0.05cm
\begin{tabular} {c|c}
  $x \in \X$ & $f(x)$ \\ \hline
  0 0 & 1 0 \\
  0 1 & 1 0 \\
  1 0 & 1 1 \\
  1 1 & 2 1 \\
  2 0 & 1 1 \\
  2 1 & 2 1 \\
\end{tabular}
\qquad
\begin{tikzcd}[ampersand replacement=\&]
    \circled{1} \arrow[loop left, "+"] \arrow[r, "+", bend left = 30] \& \circled{2} \arrow[l, bend left = 30, "+"] \\
\end{tikzcd}}
\qquad
\qquad
\subcaptionbox{}{
\tabcolsep=0.05cm
\begin{tabular} {c|c}
    $y \in \Y$ & $F(y)$ \\ \hline
  0 0 0 & 1 0 0 \\
  0 0 1 & 1 0 0 \\
  0 1 0 & 1 0 0 \\
  0 1 1 & 1 0 0 \\
  1 0 0 & 1 0 1 \\
  1 0 1 & 1 1 1 \\
  1 1 0 & 1 0 1 \\
  1 1 1 & 1 1 1 \\
\end{tabular}
\qquad
\begin{tikzcd}[ampersand replacement=\&]
    \circled{\footnotesize{1,1}} \arrow[r, bend left = 20, "+"] \arrow[d, bend right = 20, "+"'] \& \circled{\footnotesize{2,1}} \arrow[ld, bend left = 30, "+"] \\
    \circled{\footnotesize{1,2}} \&
\end{tikzcd}}}
\caption{(a): Multivalued map on $X=\{0,1,2\}\times\{0,1\}$ with interaction graph admitting a local positive cycle.
  (b): Compatible conversion $F$ of $f$, with interaction graph that does not admit any cycle.}\label{fig:ex_local_not_preserved}
  \end{center}
\end{figure}

Consider the function $\psi\colon\Y\to\A$ defined as follows. For each $i=1,\dots,n$ and $j=1,\dots,m_i$, set 
\begin{equation}\label{eq:def_psi}
  \psi_{i,j}(y_{1,1}, \dots, y_{1,m_1}, y_{2, 1}, \dots, y_{n, m_n}) = \chi_{[j,m_i]}\left(\sum_{k=1}^{m_i}y_{i,k}\right).
\end{equation}
The map $\psi$ sends a state $y$ to the admissible state $z$ such that $\sum_{j=1}^{m_i}y_{i,j} = \sum_{j=1}^{m_i}z_{i,j}$ for each $i=1,\dots,n$.
For example, taking $n=3$, $m_1=3$, $m_2 = m_3 = 2$, and $y=(0,1,1,1,0,0,1)\in\{0,1\}^7$ we have $\psi(y)=(1,1,0,1,0,1,0)$,
i.e.\ $\psi(y)$ is the image under $\b$ of the state $(2,1,1)$.

Given a map $f\colon\X\to\X$, consider the Boolean map on $\Y$ defined by
\begin{equation*}\label{eq:def_Fb}
  F = f^b \circ \psi = \b \circ f \circ \b^{-1} \circ \psi.
\end{equation*}
If $f$ is stepwise, then $F$ verifies the hypotheses of propositions~\ref{prop:A_trap_domain} and~\ref{prop:nonadm_to_adm}
and is therefore a compatible conversion of $f$.

\begin{lemma}\label{lemma:f_to_Fb_edge}
  Consider a Boolean conversion $F$ of a map $f\colon\X\to\X$.
  Suppose that the graph $G_{f^b}(x)$ at $x\in\A$ admits an edge $(j,k) \rightarrow (i,h)$ of sign $s$.
  Let $y \in \Y$ be such that $\psi(y)=x$, and $y_{j,t}=x_{j,k}$ for some $t \in \{1,\dots,m_j\}$.
  Then the graph $G_{F}(y)$ admits an edge from $(j,t)$ to $(i,h)$ of sign $s$.
\end{lemma}
\begin{proof}
  Given $x \in \A$, $i\in\{1,\dots,n\}$ and $k, t \in \{1,\dots,m_i\}$, any state $y \in \Y$ with $\psi(y) = x$ and such that $y_{j,t}=x_{j,k}$
  verifies $\psi(\overline{y}^{j,t}) = \bar{x}^{j,k}$. Therefore we have that
  \begin{align*}
    \frac{F_{i,h}(\bar{y}^{j,t})-F_{i,h}(y)}{\bar{y}^{j,t}_{j,t}-y_{j,t}} & = \frac{f^b_{i,h}(\psi(\bar{y}^{j,t}))-f^b_{i,h}(\psi(y))}{\bar{y}^{j,t}_{j,t}-y_{j,t}}
                                                                              = \frac{f^b_{i,h}(\bar{x}^{j,k})-f^b_{i,h}(x)}{\bar{x}^{j,k}_{j,k}-x_{j,k}}.
  \end{align*}
\end{proof}

Notice that, for each $x\in\Y$, $j\in\NN{n}$ and $k,t\in\NN{m_j}$, we can always find
a state $y\in\Y$ with $y_{j,t}=x_{j,k}$ and $\psi(y)=\psi(x)$,
by setting $y_{j,t}=x_{j,k}$, $y_{j.k}=x_{j,t}$, and $y_{i,h}=x_{i,h}$ for any other pair of indices.

We are now ready to show that, if the interaction graph of the original multivalued map $f$
has a cycle, then a corresponding cycle appears in the interaction graph of the Boolean conversion $F=f^b\circ\psi$.

\begin{theorem}\label{thm:circuits_f_to_F}
  Consider a map $f\colon\X\to\X$, and define $F=f^b\circ\psi\colon\Y\to\Y$.
  \begin{enumerate}[(i)]
    \item If, for some $x\in\X$, the graph $G_f(x)$ admits a positive (resp. negative) cycle, then
          there exists $y\in\Y$ with $\psi(y)=\b(x)$ such that $G_{F}(y)$
          admits a positive (resp. negative) cycle.
    \item If the global interaction graph of $f$ admits a positive (resp. negative) cycle, then
          the global interaction graph of $F$ admits a positive (resp. negative) cycle.
  \end{enumerate}
\end{theorem}
\begin{proof}
  Consider the case of a local cycle $i_1 \xrightarrow{s_1} i_2 \rightarrow \dots \rightarrow i_k \xrightarrow{s_k} i_1$
  at $G_f(x)$, $x \in \X$.
  By Lemma~\ref{lem:int_graph_f_to_fb}, the graph $G_{f^b}(\b(x))$ admits edges
  \[(i_1,h_1)\xrightarrow{s_1}(i_2,j_2), \ (i_2,h_2)\xrightarrow{s_2}(i_3,j_3), \ \dots, \ (i_{k-1},h_{k-1})\xrightarrow{s_{k-1}}(i_k,j_k), \ (i_k,h_k)\xrightarrow{s_k}(i_1,j_1).\]
  Take $y\in\Y$ such that $\psi(y)=b(x)$ and $y_{i_1,j_1}=b(x)_{i_1,h_1}$,\dots,$y_{i_k,j_k}=b(x)_{i_k,h_k}$.
  By Lemma~\ref{lemma:f_to_Fb_edge} the graph $G_{F}(y)$ contains the cycle
  \[(i_1,j_1)\xrightarrow{s_1}(i_2,j_2)\xrightarrow{s_2}(i_3,j_3)\rightarrow \dots \rightarrow(i_{k-1},j_{k-1})\xrightarrow{s_{k-1}}(i_k,j_k)\xrightarrow{s_k}(i_1,j_1).\]
  The proof for the case of a cycle in the global interaction graph proceeds similarly.
\end{proof}

If a local interaction graph of the multivalued map contains a cycle,
the interaction graph at the corresponding admissible state does not necessarily contain a cycle;
however, a local graph at a state mapped to the admissible state by $\psi$ contains a cycle.
This means that cycles in the multivalued map might not have a corresponding cycle in the graph of $f^b$,
but do have a corresponding cycle in the interaction graph of $F$.
For example, for the map in Figure~\ref{fig:ex_local_not_preserved_a}, for $F=f^b\circ\psi$ we have $F(0,1,0)=(1,0,1)$ and $F(0,1,1)=(1,1,1)$;
the loop $(1,2)\to (1,2)$ is contained in $G_F(0,0,1)$ and $G_F(0,1,1)$, and the cycle $(1,2)\to (2,1)\to (1,2)$
is in $G_F(0,1,0)$ and $G_F(0,1,1)$.

Finally, we show that a cycle in $G_{F}$ corresponds to some cycle in $G_f$.

\begin{lemma}\label{lem:GFb}
  Consider a Boolean conversion $F$ of a map $f\colon\X\to\X$.
  For each $x\in\Y$, if the graph $G_{F}(x)$ contains an edge from a vertex in $I_{j}$ to a vertex $(i, h)\in I_{i}$ with sign $s$,
  then the graph $G_{f^b}(\psi(x))$ contains an edge from some vertex in $I_j$ to the vertex $(i, h)$, with sign $s$.
\end{lemma}
\begin{proof}
  Consider a state $x \in \Y$, and indices $(i,h) \in I_i$ and $(j,k) \in I_j$. Then we have
  \[F(\bar{x}^{j,k}) - F(x) = f^b(\psi(\bar{x}^{j,k})) - f^b(\psi(x)).\]
  It is easy to see that there exists a unique $k'\in\NN{m_j}$ such that $\psi(\bar{x}^{j,k}) = \overline{\psi(x)}^{j,k'}$,
  and $x_{j,k} = 0$ if and only if ${\psi(x)}_{j,k'} = 0$. We can therefore write
  \begin{align*}
    \frac{F_{i,h}(\bar{x}^{j,k}) - F_{i,h}(x)}{\bar{x}^{j,k}_{j,k} - x_{j,k}} & =
    \frac{f^b_{i,h}(\overline{\psi(x)}^{j,k'}) - f^b_{i,h}(\psi(x))}{\overline{\psi(x)}^{j,k'}_{j,k'} - {\psi(x)}_{j,k'}}\cdot\frac{\overline{\psi(x)}^{j,k'}_{j,k'} - {\psi(x)}_{j,k'}}{\bar{x}^{j,k}_{j,k} - x_{j,k}} \\
    & = \frac{f^b_{i,h}(\overline{\psi(x)}^{j,k'}) - f^b_{i,h}(\psi(x))}{\overline{\psi(x)}^{j,k'}_{j,k'} - {\psi(x)}_{j,k'}},
  \end{align*}
  which concludes the proof.
\end{proof}

\begin{theorem}\label{thm:circuits_F_to_f}
  Consider a map $f\colon\X\to\X$, and define $F=f^b\circ\psi\colon\Y\to\Y$.
  \begin{enumerate}[(i)]
    \item If, for some $y\in\Y$, the graph $G_{F}(y)$ contains a cycle,
          then the interaction graph $G_f(x)$ contains a cycle, where $x=\b^{-1}(\psi(y))$.
          If the cycle in $G_{F}(y)$ is negative, then $G_f(x)$ contains a negative cycle.
    \item If the global interaction graph of $F$ contains a cycle,
          then the global interaction graph of $f$ contains a cycle.
          If the cycle in $G_{F}$ is negative, then $G_f$ contains a negative cycle.
  \end{enumerate}
\end{theorem}
\begin{proof}
  If $(i_1, h_1)\xrightarrow{s_1}\cdots\xrightarrow{s_{k-1}} (i_k, h_k)\xrightarrow{s_k} (i_1, h_1)$ is a cycle
  in $G_{F}(y)$ for some state $y\in\Y$, then by Lemma~\ref{lem:GFb} the graph $G_{f^b}(\psi(y))$ contains edges
  \[(i_1, h_1') \xrightarrow{s_1} (i_2, h_2),\ (i_2, h_2') \xrightarrow{s_2} (i_3, h_3),\ \dots,\ (i_{k}, h_{k}') \xrightarrow{s_k} (i_1, h_1).\]
  From Lemma~\ref{lem:int_graph_fb_to_f}, we find that the graph $G_f(x)$ with $x=\b^{-1}(\psi(y))$ contains the path
  \[i_1 \xrightarrow{s_1} i_2  \xrightarrow{s_2} \cdots \xrightarrow{s_{k-1}} i_{k} \xrightarrow{s_k} i_1.\]
  The proof for the second part is similar.
\end{proof}

The conclusion on the sign of the cycle does not hold for positive cycles:
a (local) positive cycle in $G_{F}$ could correspond to a (local) non-elementary positive cycle in $G_f$, as in the following example.

\begin{example}\label{ex:pos_no_preserved}
Consider the multivalued map with $n=3$, $m_1=m_3=1$, $m_2=2$ defined as follows:
$f(0,2,0)=f(1,2,0)=(0,1,0)$, 
$f(0,2,1)=f(1,1,0)=(1,1,0)$, 
$f(1,0,0)=f(1,0,1)=(1,1,1)$, 
$f(1,1,1)=f(1,2,1)=(1,2,0)$, 
and $f(x)=(1,0,0)$ otherwise.
It can be verified that the local interaction graphs of $f$ do not admit positive cycles
that involve more than one variable.
The local interaction graph at $(1,1,0)$ consists of two negative cycles, with vertex $2$ in common:
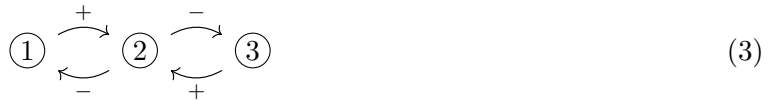
\begin{equation}\label{eq:two_neg}
  \begin{tikzcd}[column sep=scriptsize,row sep=tiny]
    \circled{1} \arrow[r, bend left=30, "+"] & \circled{2} \arrow[r, bend left=30, "-"] \arrow[l, bend left=30, "-"] & \circled{3} \arrow[l, bend left=30, "+"]
  \end{tikzcd}
\end{equation}
In the conversion, two variables are introduced to represent the second component.
The interaction graphs $G_F(1,1,0,0)$ and $G_F(1,0,1,0)$ are respectively
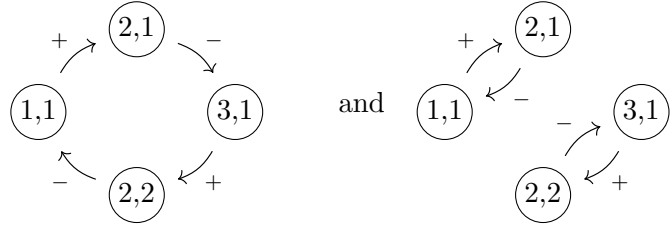
\begin{equation*}
  \begin{tikzcd}[column sep=tiny,row sep=0.2em]
    & \circled{2,1}\arrow[rd, bend left=20, "-"] & \\
    \circled{1,1} \arrow[ru, bend left=20, "+"] & & \circled{3,1} \arrow[ld, bend left=20, "+"] \\
    & \circled{2,2} \arrow[lu, bend left=20, "-"] &
  \end{tikzcd}\hspace{20pt}\text{ and }\
  \begin{tikzcd}[column sep=tiny,row sep=0.2em]
    & \circled{2,1}\arrow[ld, bend left=20, "-"] & \\
    \circled{1,1} \arrow[ru, bend left=20, "+"] & & \circled{3,1} \arrow[ld, bend left=20, "+"] \\
    & \circled{2,2} \arrow[ru, bend left=20, "-"] &
  \end{tikzcd}
\end{equation*}
The interaction graph at $(1,1,0,0)$ consists therefore of a positive cycle that involves all four variables,
and corresponds to a non-elementary cycle found at $G_f(1,1,0)$
as the composition of the two negative cycles in~\eqref{eq:two_neg}.
Two negative cycles corresponding to the two cycles in $G_f(1,1,0)$ appear instead in $G_F(1,0,1,0)$.
\end{example}

Notice that one can consider the Boolean conversion $F=f^b\circ\psi$ for maps $f$ that are not stepwise,
obtaining the preservation of the interaction cycles, but not necessarily the preservation of the attractors.
\begin{example}\label{ex:not_stepwise}
  Consider the non-stepwise map $f$ on $\{0,1,2\}^2$ defined by $f(1,1)=(1,2)$, $f(1,2)=(2,0)$,
  $f(2,2)=(0,2)$ and $f(x)=(0,0)$ otherwise.
  The asynchronous state transition graphs for $f$ and the Boolean conversion defined by $F=f^b\circ\psi$ are as follows:

  \medskip

  \begin{minipage}{4cm}
  
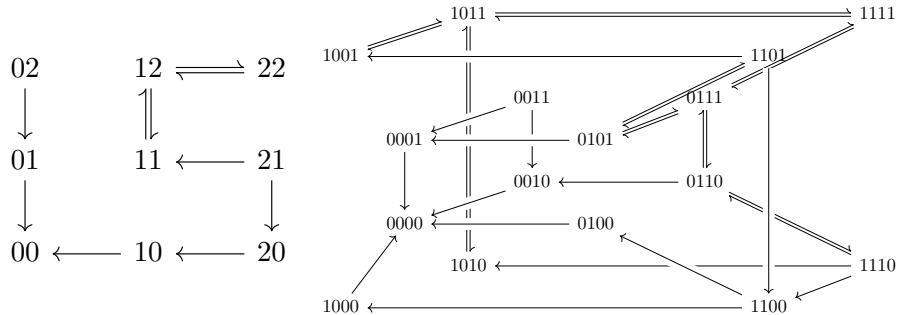
\centering
    \begin{tikzcd}%[column sep=small,row sep=small]
      02 \arrow[d] & 12 \arrow[d,rightharpoonup,xshift=+1pt] \arrow[r,rightharpoonup,yshift=+1pt] & 22 \arrow[l,rightharpoonup,yshift=-1pt] \\
      01 \arrow[d] & 11 \arrow[u,rightharpoonup,xshift=-1pt] & 21 \arrow[l] \arrow[d] \\
      00 & 10 \arrow[l] & 20 \arrow[l]
    \end{tikzcd}
  \end{minipage}
    \begin{minipage}{8cm}
    \resizebox{1.0\linewidth}{!}{
    \begin{tikzcd}[ampersand replacement=\&,row sep=small,column sep=tiny]
           \&      \& 1011 \&      \&      \& {\color{white}00} \&      \&      \& {\color{white}00} \& 1111 \\
      1001 \&      \&      \&      \&      \& {\color{white}00} \&      \& 1101 \& {\color{white}00} \& \\
           \&      \&      \& 0011 \&      \& {\color{white}00} \& 0111 \&      \& {\color{white}00} \& \\
           \& 0001 \&      \&      \& 0101 \& {\color{white}00} \&      \&      \& {\color{white}00} \& \\
           \&      \&      \& 0010 \&      \& {\color{white}00} \& 0110 \&      \& {\color{white}00} \& \\
           \& 0000 \&      \&      \& 0100 \& {\color{white}00} \&      \&      \& {\color{white}00} \& \\
           \&      \& 1010 \&      \&      \& {\color{white}00} \&      \&      \& {\color{white}00} \& 1110 \\
      1000 \&      \&      \&      \&      \& {\color{white}00} \&      \& 1100 \& {\color{white}00} \&
      \arrow[from=llllllluuuuuuu,to=uuuuuuu,rightharpoonup,yshift=+1pt] % 1011 -> 1111
      \arrow[from=uuuuuuu,to=llllllluuuuuuu,rightharpoonup,yshift=-1pt] % 1111 -> 1011
      \arrow[from=llllllluuuuuuu,to=lllllllu,rightharpoonup,xshift=+1pt] % 1011 -> 1010
      \arrow[from=lllllllu,to=llllllluuuuuuu,rightharpoonup,xshift=-1pt] % 1010 -> 1011
      \arrow[from=llllllluuuuuuu,to=llllllllluuuuuu,rightharpoonup,yshift=-1pt] % 1011 -> 1001
      \arrow[from=llllllllluuuuuu,to=llllllluuuuuuu,rightharpoonup,yshift=+1pt] % 1001 -> 1011
      \arrow[from=u,to=llluuu,rightharpoonup,yshift=-1pt] % 1110 -> 0110
      \arrow[from=llluuu,to=u,rightharpoonup,yshift=+1pt] % 0110 -> 1110
      \arrow[from=u,to=lllllllu] % 1110 -> 1010
      \arrow[from=u,to=ll] % 1110 -> 1100
      \arrow[from=llluuuuu,to=uuuuuuu,rightharpoonup,yshift=+1pt] % 0111 -> 1111
      \arrow[from=uuuuuuu,to=llluuuuu,rightharpoonup,yshift=-1pt] % 1111 -> 0111
      \arrow[from=llluuuuu,to=llllluuuu,rightharpoonup,yshift=-3pt] % 0111 -> 0101
      \arrow[from=llllluuuu,to=llluuuuu,rightharpoonup,yshift=-1pt] % 0101 -> 0111
      \arrow[from=llluuuuu,to=llluuu,rightharpoonup,xshift=+1pt] % 0111 -> 0110
      \arrow[from=llluuu,to=llluuuuu,rightharpoonup,xshift=-1pt] % 0110 -> 0111
      \arrow[from=llllluuuu,to=lluuuuuu,rightharpoonup,yshift=+3pt] % 0101 -> 1101
      \arrow[from=lluuuuuu,to=llllluuuu,rightharpoonup,yshift=+1pt] % 1101 -> 0101
      \arrow[from=ll,to=llllluu,crossing over] % 1100 -> 0100
      \arrow[from=ll,to=lllllllll,crossing over] % 1100 -> 1000
      \arrow[from=lllllllluuuu,to=lllllllluu,crossing over] % 0001 -> 0000
      \arrow[from=lllllluuu,to=lllllllluu,crossing over] % 0010 -> 0000
      \arrow[from=llllluu,to=lllllllluu,crossing over] % 0100 -> 0000
      \arrow[from=lllllllll,to=lllllllluu,crossing over] % 1000 -> 0000
      \arrow[from=lllllluuuuu,to=lllllluuu,crossing over] % 0011 -> 0010
      \arrow[from=lllllluuuuu,to=lllllllluuuu,crossing over] % 0011 -> 0001
      \arrow[from=llllluuuu,to=lllllllluuuu,crossing over] % 0101 -> 0001
      \arrow[from=llluuu,to=lllllluuu,crossing over] % 0110 -> 0010
      \arrow[from=lluuuuuu,to=llllllllluuuuuu,crossing over] % 1101 -> 1001
      \arrow[from=lluuuuuu,to=ll,crossing over] % 1101 -> 1100
    \end{tikzcd}
    }
    \end{minipage}

  \medskip

    The multivalued asynchronous dynamics has two attractors, given by the sets $\{(0,0)\}$ and $\{(1,1),(1,2),(2,2)\}$.
    In the Boolean conversion, only the fixed point remains, since the asynchronous dynamics contains a path
    from $(1,1,1,1)$ to $(0,0,0,0)$ through non-admissible states.
    Since $f$ has two attractors, a local interaction graph of $f$ contains a cycle of positive sign~\cite{richard2007necessary}
    and, by Theorem~\ref{thm:circuits_f_to_F}, a local positive cycle also appears in the interaction graph of $F$.
\end{example}

\section{First application: a Boolean map with a cyclic attractor and no local negative cycles for $n=6$}\label{sec:ex6}

In~\cite{richard2010negative}, Richard presented an example of discrete multivalued map with a unique cyclic attractor and no negative cycles in the local interaction graphs.
In this section we present a Boolean conversion of this map which demonstrates that the absence of local negative cycles does not imply the existence of a fixed point or the absence of cyclic attractors, for Boolean networks with $n \geq 6$.
An alternative conversion of the same counterexample can be found in~\cite{faure2018circuit}; in Section~\ref{sec:comparison} we compare the two approaches.

\begin{figure}[t]
  \resizebox{\columnwidth}{!}{
  \subcaptionbox{\label{fig:discrete_ex:a}}{
  \begin{tikzcd}[ampersand replacement=\&]
    03 \arrow[r] \& 13 \arrow[r] \& 23 \arrow[r] \& 33 \arrow[d] \\
    02 \arrow[u] \& 12 \arrow[l] \arrow[u] \& 22 \arrow[r] \arrow[u] \& 32 \arrow[d] \\
    01 \arrow[u] \& 11 \arrow[l] \arrow[d] \& 21 \arrow[r] \arrow[d] \& 31 \arrow[d] \\
    00 \arrow[u] \& 10 \arrow[l] \& 20 \arrow[l] \& 30 \arrow[l]
  \end{tikzcd}
  \vspace{20pt}
  }
  \subcaptionbox{\label{fig:discrete_ex:b}}{
  \begin{tikzcd}[ampersand replacement=\&,row sep=tiny]
      G_f(1,1), G_f(2,2): \& \circled{1} \arrow[loop left,"+"] \arrow[r,"-"] \& \circled{2} \arrow[loop right,"+"] \\
      G_f(1,2), G_f(2,1): \& \circled{1} \arrow[loop left,"+"] \& \circled{2} \arrow[loop right,"+"] \arrow[l,"+"'] \\
      G_f(1,0), G_f(2,3): \& \circled{1} \arrow[loop left,"+"] \arrow[r,"-"] \& \circled{2} \\
      G_f(0,2), G_f(3,1): \& \circled{1} \& \circled{2} \arrow[loop right,"+"] \arrow[l,"+"'] \\
      G_f(0,0), G_f(0,1), G_f(3,2), G_f(3,3): \& \circled{1} \arrow[r,"-"] \& \circled{2} \arrow[loop right,"+"] \\
      G_f(0,3), G_f(1,3), G_f(2,0), G_f(3,0): \& \circled{1} \arrow[loop left,"+"] \& \circled{2} \arrow[l,"+"'] \\
  \end{tikzcd}}}
  \caption{(a): Asynchronous state transition graph for the stepwise version $f: \{0,1,2,3\}^2 \rightarrow \{0,1,2,3\}^2$
    of a map with one cyclic attractor (Example $6$ in~\cite{richard2010negative}).
    (b): Local interaction graphs of $f$.}\label{fig:discrete_ex}
\end{figure}
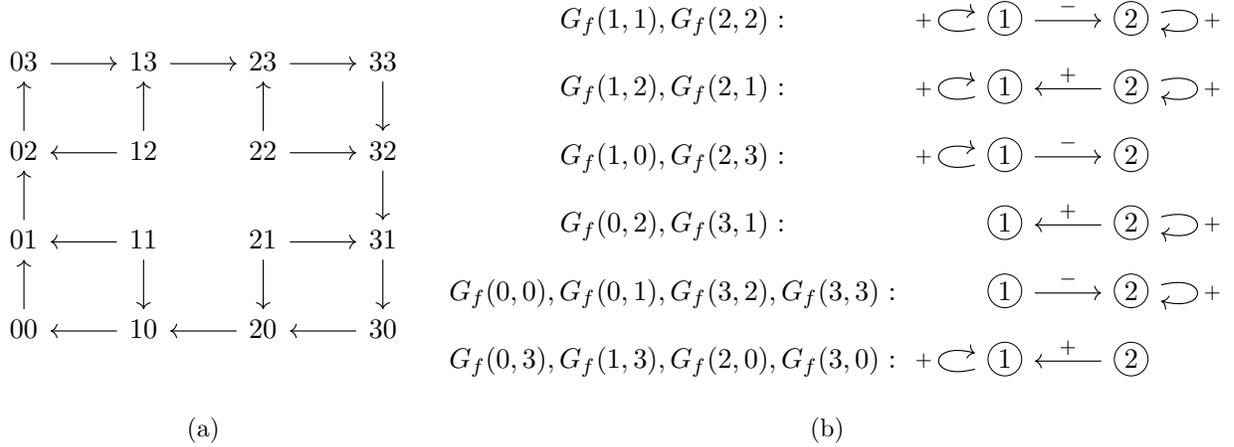

First, to create a Boolean conversion of Richard's example with a unique cyclic attractor,
we can take the stepwise version $f$, which is given in Figure~\ref{fig:discrete_ex:a}.
Since taking the stepwise version can alter the interaction graphs, we check that the local interaction graphs of $f$
do not contain any negative loop (the graphs are given in Figure~\ref{fig:discrete_ex:b}).

We then consider the Boolean conversion $F=\b\circ f\circ\b^{-1}\circ\psi$ of $f$, as described in Section~\ref{sec:circuits}.
For instance, the non-admissible state $(0,1,0,1,0,1)$ is mapped to the image under $\b\circ f\circ\b^{-1}$ of
the corresponding admissible state $\psi(0,1,0,1,0,1)=(1,0,0,1,1,0)$, i.e. to $(0,0,0,1,1,1)$.

Since $F$ is a compatible conversion of $f$ (see the end of Section~\ref{sec:asympt}),
the attractors of $AD_F$ are isomorphic to the attractors of $AD_f$, and therefore the dynamics has one attractor,
which is the image under $\b$ of of the attractive cycle of $f$ (Figure~\ref{fig:conversion:a}).
By application of Theorem~\ref{thm:circuits_F_to_f}, we find that the local interaction graphs of $F$ have no negative cycles.
The global interaction graph for $F$ takes the form given in Figure~\ref{fig:conversion:b}.

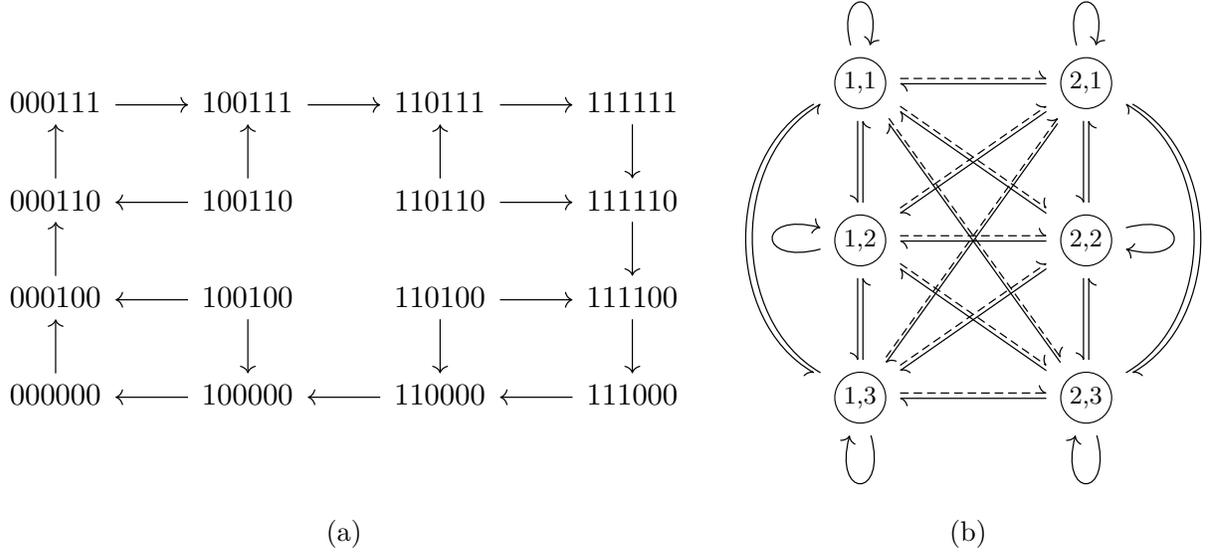
\begin{figure}
  \resizebox{\columnwidth}{!}{
  \subcaptionbox{\label{fig:conversion:a}}{
  \begin{tikzcd}[ampersand replacement=\&]
    000111 \arrow[r] \& 100111 \arrow[r] \& 110111 \arrow[r] \& 111111 \arrow[d] \\
    000110 \arrow[u] \& 100110 \arrow[l] \arrow[u] \& 110110 \arrow[r] \arrow[u] \& 111110 \arrow[d] \\
    000100 \arrow[u] \& 100100 \arrow[l] \arrow[d] \& 110100 \arrow[r] \arrow[d] \& 111100 \arrow[d] \\
    000000 \arrow[u] \& 100000 \arrow[l] \& 110000 \arrow[l] \& 111000 \arrow[l]
  \end{tikzcd}
  \vspace{30pt}}
  \subcaptionbox{\label{fig:conversion:b}}{
  \begin{tikzcd}[ampersand replacement=\&,column sep=huge,row sep=large]
      \circled{\footnotesize{1,1}} \arrow[loop above] \arrow[d, rightharpoondown, xshift=-0.2ex] \arrow[dd, rightharpoondown, bend right=60, xshift=-0.5ex, near end]
      \arrow[r, rightharpoonup, dashed, yshift=+0.2ex] \arrow[r, leftharpoondown, yshift=-0.2ex]
      \arrow[rd, rightharpoonup, dashed, yshift=+0.2ex] \arrow[rd, leftharpoondown, yshift=-0.2ex]
      \arrow[rdd, rightharpoonup, dashed, yshift=+0.2ex] \arrow[rdd, leftharpoondown, yshift=-0.2ex]
      \&
      \circled{\footnotesize{2,1}} \arrow[loop above] \arrow[d, rightharpoondown, xshift=-0.2ex] \arrow[dd, rightharpoonup, bend left=60, xshift=+0.5ex, near end] \\
      \circled{\footnotesize{1,2}} \arrow[loop left] \arrow[u, rightharpoondown, xshift=+0.2ex] \arrow[d, rightharpoondown, xshift=-0.2ex]
      \arrow[ru, rightharpoonup, dashed, yshift=+0.2ex] \arrow[ru, leftharpoondown, yshift=-0.2ex]
      \arrow[r, rightharpoonup, dashed, yshift=+0.2ex] \arrow[r, leftharpoondown, yshift=-0.2ex]
      \arrow[rd, rightharpoonup, dashed, yshift=+0.2ex] \arrow[rd, leftharpoondown, yshift=-0.2ex]
      \&
      \circled{\footnotesize{2,2}} \arrow[loop right] \arrow[u, rightharpoondown, xshift=+0.2ex] \arrow[d, rightharpoondown, xshift=-0.2ex] \\
      \circled{\footnotesize{1,3}} \arrow[loop below] \arrow[u, rightharpoondown, xshift=+0.2ex] \arrow[uu, rightharpoondown, bend left=60, near end]
      \arrow[ruu, rightharpoonup, dashed, yshift=+0.2ex] \arrow[ruu, leftharpoondown, yshift=-0.2ex]
      \arrow[ru, rightharpoonup, dashed, yshift=+0.2ex] \arrow[ru, leftharpoondown, yshift=-0.2ex]
      \arrow[r, rightharpoonup, dashed, yshift=+0.2ex] \arrow[r, leftharpoondown, yshift=-0.2ex]
      \&
      \circled{\footnotesize{2,3}} \arrow[loop below] \arrow[u, rightharpoondown, xshift=+0.2ex] \arrow[uu, rightharpoonup, bend right=60, near end]
  \end{tikzcd}}}
\caption{(a): Asynchronous state transition graph on the admissible states for the Boolean conversion $F=f^b\circ\psi$ for the map in Figure~\ref{fig:discrete_ex}. (b): Global interaction graph of $F$ (negative edges are dashed).}\label{fig:conversion}
\end{figure}

\section{Second application: mirror states and local cycles}\label{sec:mirror}

To introduce this additional application, we need some new definitions.
For $x,y\in\X$, we write $I(x,y)=\{i\in\NN{n}\ | \ x_i\neq y_i\}$.
Given a map $f\colon Y\to Y$, $Y\subseteq X$, we call a pair of distinct states $x,y\in Y$ a \emph{mirror pair} for $f$ if
\begin{equation*}
  f_i(x) \leq x_i,y_i\leq f_i(y) \ \text{or} \ f_i(y) \leq x_i,y_i\leq f_i(x) \ \text{for all } i \in I(x,y).
\end{equation*}
In the Boolean case, the condition reduces to $f_i(x)\neq f_i(y)$ for all $i\in I(x,y)$.

Ruet~\cite{ruet2016local} proved that, if a mirror pair exists for a Boolean map $f$,
then two local interaction graphs of $f$ contain a cycle.

\begin{theorem}(\cite[Theorem 5.1]{ruet2016local}) \label{thm:two_circuits}
  If $f\colon\Y\to\Y$ admits a mirror pair, then there exist two different states $x,y\in\Y$
  such that $G_f(x)$ and $G_f(y)$ contain a cycle, and $x,y$ is a mirror pair for $f$.
\end{theorem}

\begin{lemma}\label{lemma:mirror}
  The pair $x,y\in\X$ is a mirror pair for $f\colon\X\rightarrow\X$ if and only if $\b(x)$, $\b(y)$ is a mirror pair
  for $f^b=\b\circ f\circ\b^{-1}\colon\A\rightarrow\A$.
\end{lemma}
\begin{proof}
  Suppose that $x,y$ is a mirror pair for $f$, and take $(i,j) \in I(\b(x),\b(y))$.
  Then $i\in I(x,y)$, and we can assume $x_i<y_i$.
  In this case, $\b_{i,j}(x)\neq\b_{i,j}(y)$ implies $x_i<j\leq y_i$,
  and either $f_i(x)\leq x_i<j\leq y_i \leq f_i(y)$ or $f_i(y)\leq x_i<j\leq y_i \leq f_i(x)$.
  In both cases, we have $f^b_{i,j}(\b(x))\neq f^b_{i,j}(\b(y))$ as required.

  Conversely, suppose that $\b(x)$, $\b(y)$ is a mirror pair for $f^b$, and take $i\in I(x,y)$.
  We can assume again that $x_i<y_i$.
  For any $j$ such that $x_i<j\leq y_i$ we have $f^b_{i,j}(\b(x))\neq f^b_{i,j}(\b(y))$,
  which means that either $f_i(x)<j$ and $f_i(y)\geq j$ or $f_i(y)<j$ and $f_i(x)\geq j$.
  Taking $j=x_i+1$ and $j=y_i$ we find $f_i(x)\leq x_i<y_i$ and $f_i(y)\geq y_i>x_i$
  or $f_i(y)\leq x_i<y_i$ and $f_i(x)\geq y_i>x_i$.
\end{proof}

A mirror pair for $f^b$ is also a mirror pair for any Boolean conversion $F$, since 
$F$ coincides with $f^b$ on the admissible states $\A$.
Using the conversion introduced in Section~\ref{sec:circuits}, we can now prove a multivalued version of Theorem~\ref{thm:two_circuits}.

\begin{theorem}\label{thm:mirror_multi}
  If $f\colon\X\rightarrow\X$ admits a mirror pair, then there exist two different states $x,y\in\X$
  such that $G_f(x)$ and $G_f(y)$ contain a cycle.
\end{theorem}
\begin{proof}
  If $f$ admits a mirror pair, then by Lemma~\ref{lemma:mirror} the map $F=\b\circ f\circ\b^{-1}\circ\psi\colon\Y\to\Y$ admits a mirror pair.
  Theorem~\ref{thm:two_circuits} then gives the existence of two states $x'\neq y' \in \Y$ such that $G_{F}(x')$ and $G_{F}(y')$ contain a cycle.
  Taking $x=\b^{-1}(\psi(x'))$ and $y=\b^{-1}(\psi(y'))$,
  by Theorem~\ref{thm:circuits_F_to_f}, we find that the graphs $G_f(x)$ and $G_f(y)$ contain a cycle.
  In addition, since $x',y'$ is a mirror pair, $F(x')\neq F(y')$ and we have that $\psi(x')\neq\psi(y')$.
  Therefore $x$ and $y$ are distinct.
\end{proof}

We conjecture that, as for the Boolean case, under the hypotheses of Theorem~\ref{thm:mirror_multi},
there exists a mirror pair $x,y\in\X$ such that $G_f(x)$ and $G_f(y)$ contain a cycle
(this is a consequence of Lemma~\ref{lemma:mirror} when the mirror pair $x',y'$ of the proof
of Theorem~\ref{thm:mirror_multi} can be taken in the set of admissible states $\A$).

\section{Comparison with the mapping of Faur{\'e} and Kaji}\label{sec:comparison}

Recently Faur{\'e} and Kaji~\cite{faure2018circuit} proposed an alternative Boolean conversion of multivalued dynamics.
Write $\psi^*=\b^{-1}\circ \psi$, then the \emph{binarisation} $\B(f)\colon\Y\to\Y$ of a map $f\colon\X\to\X$ is defined as
\begin{equation*}
  \mathbb{B}(f)_{i,j}(x) =
    \begin{cases}
      0 & \mbox{if } f_i(\psi^*(x)) < \psi_i^*(x), \\
      x_{i,j} & \mbox{if } f_i(\psi^*(x)) = \psi_i^*(x), \\
      1 & \mbox{if } f_i(\psi^*(x)) > \psi_i^*(x).
    \end{cases}
\end{equation*}
For comparison, the conversion introduced in Section~\ref{sec:circuits} writes as $F=\b\circ f\circ\psi^*$.
Properties of the binarisation are studied in~\cite{faure2018circuit} for maps that are \emph{asymptotic}.
The \emph{asymptotic} version $\hat{f}\colon\X\to\X$ of $f$ is defined as follows:
\begin{equation*}\label{eq:asymptotic}
  \hat{f}_i(x) =
    \begin{cases}
      0 & \mbox{if } f_i(x) < x_i, \\
      x_i & \mbox{if } f_i(x) = x_i, \\
      m_i & \mbox{if } f_i(x) > x_i,
    \end{cases}
\end{equation*}
and a map $f$ is called \emph{asymptotic} if it coincides with its asymptotic version.
The interaction graphs of the asymptotic and stepwise version of a map $f$
can differ in terms of loops, but they otherwise coincide.

The binarisation defines a bijection from the set of asymptotic maps on $\X$ to maps on $\Y$ that are symmetric with
respect to the permutations of the components that correspond to the same multivalued variable:
if $\sigma_i$ is a permutation of $\{1,\dots,m_i\}$, for $i=1,\dots,n$, then the binarisation of $f$ satisfies $\B(f)_{i,j}(x)=\B(f)_{i,\sigma_i(j)}(x_{1,\sigma_1^{-1}(1)},x_{1,\sigma_1^{-1}(2)},\dots,x_{n,\sigma_n^{-1}(m_n)})$.
This symmetry property does not hold for the conversion presented in this work (consider for instance the map in Figure~\ref{fig:ex_binarisation}).

\textit{Asymptotic behaviour.} For a compatible conversion $F$ of $f$,
the attractors of the asynchronous dynamics are in one-to-one correspondence with the attractors of $f$.
In contrast, the attractors of the binarisation $\B(f)$ of an asymptotic function $f$ map surjectively onto the attractors of $AD_f$.
The dynamics $AD_{\B(f)}$ admits a fixed point if and only if $AD_f$ admits a fixed point,
and it admits a cyclic attractor if and only if $AD_f$ admits a cyclic attractor.
Consider for instance the map on $\{0,1,2\}$ defined by $f(x)=2-x$, whose asynchronous state transition state transition graph is represented in Figure~\ref{fig:ex_binarisation}:
the compatible conversion $F=f^b\circ\psi$ has one fixed point $(1,0)$, whereas the binarisation $\B(f)$ admits two symmetric fixed points $(1,0)$ and $(0,1)$.

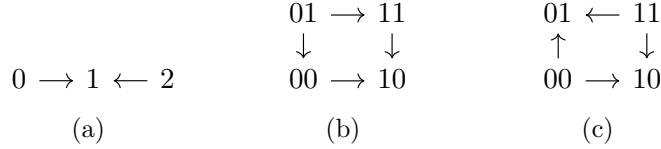
\begin{figure}
  \begin{center}
    \subcaptionbox{}{
      \begin{minipage}{3cm}
      \centering
      \begin{tikzcd}[ampersand replacement=\&,row sep=small,column sep=small]
        0 \arrow[r] \& 1 \& 2 \arrow[l]
      \end{tikzcd}
      \end{minipage}
      }
    \subcaptionbox{}{
      \begin{minipage}{3cm}
      \centering
      \begin{tikzcd}[ampersand replacement=\&,row sep=small,column sep=small]
        01 \arrow[r] \arrow[d] \& 11 \arrow[d] \\
        00 \arrow[r] \& 10
      \end{tikzcd}
      \end{minipage}
      }
    \subcaptionbox{}{
      \begin{minipage}{3cm}
      \centering
      \begin{tikzcd}[ampersand replacement=\&,row sep=small,column sep=small]
        01 \& 11 \arrow[l] \arrow[d] \\
        00 \arrow[u] \arrow[r] \& 10
      \end{tikzcd}
      \end{minipage}
    }
  \end{center}
\caption{(a): Asynchronous dynamics of the map $f\colon\{0,1,2\}\to\{0,1,2\}$ defined by $f(x)=2-x$.
  (b): Asynchronous dynamics of the conversion $F=\b\circ \tf\circ\b^{-1} \circ\psi$ of the stepwise version $\tf$ of $f$.
  (c) Asynchronous dynamics of $\B(f)$.}\label{fig:ex_binarisation}
\end{figure}

\textit{Interaction graphs.} To an edge $i\to j$ in the interaction graph $G_f(x)$ corresponds an edge $(i,i')\to(j,j')$
of the same sign in $G_{\B(f)}(y)$, for all $y\in\Y$ with $x=\psi^*(y)$.
Conversely, an edge $(i,i')\to(j,j')$ in $G_{\B(f)}(y)$ corresponds to an edge $i\to j$ of the same sign in $G_f(\psi^*(y))$.
As a consequence, a cycle in a local interaction graph exists for $\B(f)$ only if it exists for the asymptotic multivalued map $f$.
Since the asymptotic version of the map of Richard's counterexample considered in Section~\ref{sec:ex6}
does not admit any local negative cycle, the binarisation of the counterexample has the same property.
In addition, all the attractors of the binarisation are cyclic.
In fact, the binarisation has only one attractor, which has $28$ vertices and is not a cycle.

Another difference between the two conversion methods is in the preservation of the local cycles in
the interaction graph of $f$: cycles in $G_f(x)$ do not always have a correspondence in the interaction graphs $G_{\B(f)}(y)$ with $x=\psi^*(y)$.
This property of the binarisation might be used to identify some cycles that are not relevant for the asymptotic behaviour,
since these are ``lost'' in the conversion.
For example, consider the local interaction graph $G_f(1,0)$ of the map $f$ in Figure~\ref{fig:ex_binarisation_2},
which consists of a cycle.
The local interaction graphs of $\B(f)$ at $(0,1,0)$ and $(1,0,0)$ are symmetric and do not admit any cycle,
whereas the interaction graph of the conversion $F=\b\circ f\circ\b^{-1}\circ\psi$ admits a cycle at $(0,1,0)$.

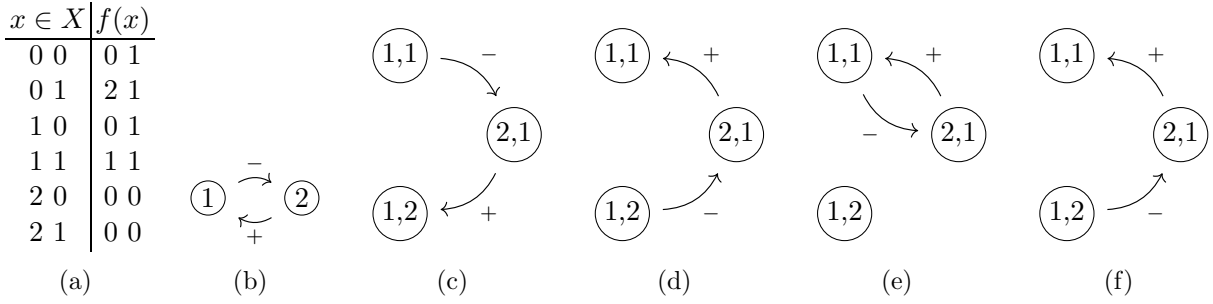
\begin{figure}
  \begin{center}
  \resizebox{\columnwidth}{!}{
    \subcaptionbox{}{
    \tabcolsep=0.05cm
    \begin{tabular} {c|c}
      $x \in \X$ & $f(x)$ \\ \hline
      0 0 & 0 1 \\
      0 1 & 2 1 \\
      1 0 & 0 1 \\
      1 1 & 1 1 \\
      2 0 & 0 0 \\
      2 1 & 0 0 \\
    \end{tabular}}
    \subcaptionbox{}{
    \begin{tikzcd}[ampersand replacement=\&,column sep=small,row sep=small]
        \circled{1} \arrow[r, "-", bend left = 30] \& \circled{2} \arrow[l, bend left = 30, "+"] \\
    \end{tikzcd}
    }
    \subcaptionbox{}{
    \begin{tikzcd}[ampersand replacement=\&,column sep=small,row sep=0.1em]
      \circled{1,1} \arrow[rd,"-",bend left=30] \& \\
      \& \circled{2,1} \arrow[ld,bend left=30,"+"] \\
      \circled{1,2} \&
    \end{tikzcd}
      }
    \subcaptionbox{}{
    \begin{tikzcd}[ampersand replacement=\&,column sep=small,row sep=0.1em]
      \circled{1,1} \& \\
      \& \circled{2,1} \arrow[lu,bend right=30,"+"'] \\
      \circled{1,2} \arrow[ru,"-"',bend right=30] \&
    \end{tikzcd}
      }
    \subcaptionbox{}{
    \begin{tikzcd}[ampersand replacement=\&,column sep=small,row sep=0.1em]
      \circled{1,1} \arrow[rd,"-"',bend right=30] \& \\
      \& \circled{2,1} \arrow[lu,bend right=30,"+"'] \\
      \circled{1,2} \&
    \end{tikzcd}
      }
    \subcaptionbox{}{
    \begin{tikzcd}[ampersand replacement=\&,column sep=small,row sep=0.1em]
      \circled{1,1} \& \\
      \& \circled{2,1} \arrow[lu,bend right=30,"+"'] \\
      \circled{1,2} \arrow[ru,"-"',bend right=30] \&
    \end{tikzcd}
      }
  }
    \caption{(a): Asymptotic map $f\colon\{0,1,2\}\times\{0,1\}\to\{0,1,2\}\times\{0,1\}$.
      (b): Local interaction graph of $f$ at $(1,0)$.
      (c) and (d): Local interaction graphs of $\B(f)$ at $(0,1,0)$ and $(1,0,0)$.
      (e) and (f): Local interaction graphs of $F=\b\circ f\circ\b^{-1}\circ\psi$ at $(0,1,0)$ and $(1,0,0)$.}\label{fig:ex_binarisation_2}
  \end{center}
\end{figure}

\section{Other interaction graphs and cycles in the admissible region}\label{sec:circuits_admissible}

In this section, we study the interaction structure of the partial Boolean conversion
$f^b=\b\circ f\circ\b^{-1}$ of a multivalued map $f$.
In particular, we find versions of the first and second conjecture of Thomas~\cite{thomas1981relation} for $f^b$.

In Section~\ref{sec:circuits}, we showed how the existence of a cycle in the interaction graphs of $f$
does not imply the existence of a cycle in the interaction graphs of $f^b$ (see Figure~\ref{fig:ex_local_not_preserved}).
We consider now a more restrictive definition of interaction graph for $f$, introduced in~\cite{richard2007necessary},
and show that cycles in these graphs have a correspondence in the interaction graphs of $f^b$.

\begin{definition}\label{def:non-usual}
  The \emph{non-usual local interaction graph} $\tG_f(x,y)$ of the map $f\colon\X\to\X$ at a state $x\in\X$
  \emph{with variations in direction of $y$} is the graph with vertex set $\{1, \dots, n\}$,
  and an edge from a vertex $j$ to a vertex $i$ of sign $s$, with $i, j \in I(x,y)$,
  whenever $x + \epsilon_j\e^j \in \X$ and $s = \epsilon_j \sign(f_i(x + \epsilon_j\e^j) - f_i(x))\neq 0$,
  with $\epsilon_k = \sign(y_k - x_k)$ for all $k \in I(x,y)$, and, in addition,
  \[\min{\{f_i(x),f_i(x + \epsilon_j\e^j)\}} < x_i + \frac{\epsilon_i}{2} < \max{\{f_i(x),f_i(x + \epsilon_j\e^j)\}}.\]
\end{definition}
Any non-usual local interaction graph at a state $x \in \X$ is clearly a subgraph of $G_f(x)$.
Given an edge in $G_f(x)$ from $j$ to $i$ with sign $s$ and variation $s_1$,
and a state $y$ with $\epsilon_j=y_j-x_j=s_1$, $y_i-x_i=\epsilon_i$, the graph $\tG_f(x,y)$ contains an edge from $j$ to $i$ with sign $s$
only if $f_i(x)$ and $f_i(x + \epsilon_j\e^j)$ are on different sides of $x+\frac{\epsilon_i}{2}$.
Intuitively, an edge $j\to i$ of the (standard) interaction graph is discarded if
moving from $x$ to $x + \epsilon_j\e^j$ does not alter the direction of change that $f_i$ imposes on component $i$.

For instance, consider the example in Figure~\ref{fig:ex_local_not_preserved}.
To find the local interaction graph $G_f(x)$ at $x=(1,0)$ we compare $f(1,0)=(1,1)$ to $f(x-\e^1)=f(0,0)=(1,0)$ finding a positive edge $1\to2$,
to $f(x+\e^1)=f(2,0)=(1,1)$ finding no edges, and to $f(x+\e^2)=f(1,1)=(2,1)$, finding a positive edge $2\to1$.
Consider now the state $y=(x_1-1,x_2+1)=(0,1)$, so that $\epsilon_1=-1$ and $\epsilon_2=+1$.
Since $f_2(x)=1$ and $f_2(x+\epsilon_1\e^1)=f_2(0,0)=0$ are on different sides of $x_2+\frac{\epsilon_2}{2}=\frac{1}{2}$,
the edge $1\to2$ is in $\tG_f(x,y)$.
We have on the other hand that $f_1(x)=1$ and $f_1(x+\epsilon_2\e^2)=f_1(1,1)=2$, which are on the same side of $x_1+\frac{\epsilon_1}{2}=\frac{1}{2}$,
hence the edge $1\to2$ is not in $\tG_f(x,y)$.
In fact, the non-usual interaction graphs of $f$ do not contain any cycle.

\begin{lemma}\label{lem:non-usual_int_graph}
    Consider $f\colon\X\to\X$ and $x,y\in\X$.
    If the non-usual interaction graph $\tG_f(x,y)$ at $x$ with variations in the direction of $y$ contains an edge from $j$ to $i$ of sign $s$,
    and $\epsilon_k = \sign(y_k-x_k)$ for $k \in I(x,y)$,
    then the graph $G_{f^b}(\b(x))$ contains an edge from
    $(j,x_j + \frac{\epsilon_j+1}{2})$ to $(i,x_i + \frac{\epsilon_i+1}{2})$, with sign $s$.
\end{lemma}
\begin{proof}
    If the non-usual local interaction graph $\tG_f(x,y)$ at $x$ contains an edge from $j$ to $i$,
    we have, by definition of non-usual interaction graph,
    \[\min{\{f_i(x),f_i(x + \epsilon_j\e^j)\}} < x_i + \frac{\epsilon_i}{2} < \max{\{f_i(x),f_i(x + \epsilon_j\e^j)\}},\]
    which gives
    \[\min{\{f_i(x),f_i(x + \epsilon_j\e^j)\}} < x_i + \frac{\epsilon_i+1}{2} \leq \max{\{f_i(x),f_i(x + \epsilon_j\e^j)\}}.\]
    The conclusion follows from Lemma~\ref{lem:int_graph_f_to_fb}.
\end{proof}

\begin{proposition}\label{prop:circuit_non-usual}
    If the non-usual interaction graph $\tG_f(x,y)$ of the map $f\colon\X\to\X$ at a state $x \in \X$ with variations in direction of $y$
    admits a cycle of sign $s$, then the graph $G_{f^b}(\b(x))$ admits a cycle of sign $s$.
\end{proposition}
\begin{proof}
    Let $i_1\to\dots\to i_{k-1}\to i_k = i_1$ be a cycle in $\tG_f(x,y)$, with edge signs $s_1, \dots, s_{k-1}$, and take $\epsilon_k = \sign(y_k - x_k)$, $k \in I(x,y)$.
    By Lemma~\ref{lem:non-usual_int_graph}, the graph $G_{f^b}(\b(x))$ contains an edge
    from $(i_h, x_h + \frac{\epsilon_h+1}{2})$ to $(i_{h+1}, x_{h+1} + \frac{\epsilon_{h+1}+1}{2})$ with sign $s_h$, for all $h = 1, \dots, k-1$, which concludes the proof.
\end{proof}

Non-usual interaction graphs are used in~\cite{richard2008extension} to prove a discrete analogue of the Jacobian conjecture,
and in~\cite{richard2007necessary} to prove a version of the first conjecture of Thomas.
We can use these theorems to prove analogous results for $f^b$.

\begin{theorem}\label{thm:sd_multi}(\cite[Theorem 1]{richard2008extension})
  If a map $f\colon\X\to\X$ is such that $\tG_f(x,y)$ has no cycles for all $x,y\in\X$, then $f$ has a unique fixed point.
\end{theorem}

\begin{theorem}\label{thm:first_conj}(\cite[Corollary 1]{richard2007necessary})
  If a map $f\colon\X\to\X$ is such that $AD_f$ admits two distinct attractors,
  then there exists a state $x\in \X$ such that $\tG_f(x,y)$ has a positive cycle for some $y\in\X$.
\end{theorem}

\begin{theorem}
  Given $f\colon\X\to\X$, suppose that the graph $G_{f^b}(x)$ does not have any cycle for all $x\in\Y$.
  Then $f^b$ has a unique fixed point.
\end{theorem}
\begin{proof}
  If the graphs $G_{f^b}(x)$, $x\in\Y$, admit no cycle, then by Proposition~\ref{prop:circuit_non-usual},
  the graphs $\tG_{f}(x',y')$ do not have any cycle, for all $x',y'\in\X$.
  The conclusion follows from Theorem~\ref{thm:sd_multi}.
\end{proof}

\begin{theorem}
  Suppose that the asynchronous dynamics of a map $f\colon\X\to\X$ admits two distinct attractors.
  Then there exists $x\in\A$ such that the graph $G_{f^b}(x)$ admits a positive cycle.
\end{theorem}
\begin{proof}
  By Theorem~\ref{thm:first_conj}, there exist states $x,y\in\X$ such that $\tG_f(x,y)$ has a positive cycle,
  and by Proposition~\ref{prop:circuit_non-usual}, the graph $G_{f^b}(\b(x))$ admits a positive cycle.
\end{proof}

From the last theorem we can also derive some conclusions on Boolean maps.
Suppose that $\A$ is a subset of $\Y$ that can be interpreted
as the admissible region for some multivalued space $X$.
If the restriction of the asynchronous dynamics to $\A$ admits multiple attractors,
then a local interaction graph must contain a positive cycle (see for instance the Boolean map in Example~\ref{ex:not_stepwise}).

To discuss a version of the second conjecture of Thomas for the partial Boolean conversion $f^b$
of a map $f\colon\X\to\X$, we introduce other definitions.
For each $i=1,\dots,n$, define a map $F^i\colon\X\to\X$ by setting
\[F^i(x)=x + \sign(f_i(x)-x_i)\e^i\]
for all $x\in\X$. That is, $F^i(x)$ coincides with $x$ except on component $i$, where it coincides with the stepwise version of $f$.
The following definition of interaction graph was introduced in~\cite{richard2010negative}, where it is used to prove a multivalued version of the second conjecture of Thomas.
\begin{definition}\label{def:tildeGf_x}(\cite[Definition 5]{richard2010negative})
    Given $f\colon\X\to\X$ and $x\in\X$, $\G_f(x)$ is the graph with vertex set $\{1, \dots, n\}$ that contains an edge from $j$ to $i$ of sign $s$ if
    \begin{itemize}
        \item[(i)] $\sign(f_i(x) - x_i) \neq \sign(f_i(F^j(x)) - F_i^j(x))$ \text{and}
        \item[(ii)] $s = \sign(f_j(x) - x_j) \sign(f_i(F^j(x)) - F_i^j(x))\neq 0$.
    \end{itemize}
\end{definition}

\begin{remark}\label{rmk:G_subgraph}
  Notice that, as a consequence of point $(ii)$ of Definition~\ref{def:tildeGf_x},
  the graph $\G_f(x)$ contains edges of $G_f(x)$ that are calculated for directions $j$ such that $f_j(x) \neq x_j$,
  i.e. such that $(x,x-\e^j)$ or $(x,x+\e^j)$ is a transition in the asynchronous dynamics of $f$.
\end{remark}

Given $A\subseteq\X$, we write $\G_f(A)$ for the graph with vertex set $\NN{n}$
and an edge from $j$ to $i$ with sign $s$ if, for some $x\in A$, the graph $\G_f(x)$
has an edge from $j$ to $i$ with sign $s$.

\begin{theorem}\label{thm:second_conj} (\cite[Theorem $2$]{richard2010negative})
  If the asynchronous dynamics of $f\colon\X\to\X$ has a cyclic attractor $A$,
  then the graph $\G_f(A)$ admits a negative cycle.
\end{theorem}

\begin{lemma}\label{lemma:G_subgraph}
    Given $f\colon\X\to\X$, for all $x\in\X$, $\G_f(x)$ is a subgraph of $G_f(x)$.
\end{lemma}
\begin{proof}
    Let $j\to i$ be an edge of $\G_f(x)$ of sign $s$. Then from point $(ii)$ of Definition~\ref{def:tildeGf_x} we have that $f_j(x) \neq x_j$ and $f_i(F^j(x)) \neq F_i^j(x)$.
    Therefore we can write $f_j(x) - x_j = s_1 k_1$, $F^j(x) = x + s_1 \e^j$ and $f_i(F^j(x)) - F_i^j(x) = s_2 k_2$, with $k_1, k_2 > 0$, $s_1, s_2 \in {\{-1, +1\}}$ and $s = s_1 s_2$.

    Moreover, from $(i)$, we find that $f_i(x) - x_i = -s_2 h_2$ for some $h_2 \geq 0$.
    To conclude that $G_f(x)$ admits an edge from $j$ to $i$ of sign $s$, we show that $\sign(f_i(x + s_1 \e^j) - f_i(x)) = s_2$.

    If $i \neq j$, we have $F_i^j(x) = x_i$ and we can write
    \[f_i(x + s_1 \e^j) - f_i(x) = f_i(x + s_1 \e^j) - F_i^j(x) + x_i - f_i(x) = s_2 (k_2 + h_2).\]

    If instead $i = j$, then necessarily $s_2 = -s_1$ and $h_2 > 0$, and
    \begin{equation*}
      \begin{aligned}
        f_i(x + s_1 \e^j) - f_i(x) & = f_i(x + s_1 \e^j) - F_i^j(x) + x_i + s_1 - f_i(x) \\
        & = s_2 (k_2 + h_2) + s_1 = s_2 (k_2 + h_2 - 1) \neq 0,
      \end{aligned}
    \end{equation*}
    which concludes the proof.
\end{proof}

\begin{theorem}\label{thm:neg_circuits_attr_fb}
  Consider a map $f\colon\X\to\X$, and suppose that $AD_f$ admits a cyclic attractor $A$.
  Then the graph $G_{f^b}(\b(A))$ admits a negative cycle.
\end{theorem}
\begin{proof}
  Consider a compatible conversion $F$ of $f$. Then $\b(A)$ is an attractor for $AD_F$ and,
  by Theorem~\ref{thm:second_conj}, $\G_F(\b(A))$ admits a negative cycle.
  As a consequence of Remark~\ref{rmk:G_subgraph} and Lemma~\ref{lemma:G_subgraph}, the cycle is contained in $G_{f^b}(\b(A))$.
\end{proof}

\section{Conclusion}

We presented a method for converting multivalued asynchronous dynamics to Boolean maps that preserves
the asymptotic behaviour and the interaction structure.
As opposed to the conversion technique of Faur{\'e} and Kaji~\cite{faure2018circuit},
for a multivalued map $f$, we focused on extending the partial Boolean conversion $f\mapsto f^b=\b\circ f\circ\b^{-1}$ obtained using a natural embedding
of multivalued states into Boolean states, introduced and studied in~\cite{van1979deal,didier2011mapping},
to find a one-to-one mapping between attractors of the multivalued and Boolean asynchronous dynamics.
Using the conversion, we were able to construct a $6$-dimensional
Boolean map with an attractive cycle, no fixed points and no local negative cycles~\ref{sec:ex6},
and we showed that multivalued maps with mirror states have local cycles~\ref{sec:mirror}.

In the last part of the work, we looked more closely at the interaction graphs of the map $f^b$,
that can be defined without considering an explicit extension to the non-admissible states,
and observed that results connecting cycles and asymptotic behaviour still hold for these smaller graphs.
The analysis demonstrates how observations on multivalued dynamics can translate into results on restrictions of Boolean maps to admissible regions,
and, conversely, statements on Boolean interaction networks that hold on subregions that are embeddings of multivalued spaces
can be sometimes transported to the multivalued case.

\section*{Acknowledgements}
The author is grateful to Claudine Chaouiya, Etienne Farcot, Adrien Faur{\'e} and Shizuo Kaji for useful discussions,
as well as to the anonymous reviewers for their very helpful comments.

\bibliographystyle{plainnat}%plain
\bibliography{biblio}

\end{document}